\documentclass[conference]{IEEEtran}
\usepackage{CJK}
\usepackage{enumerate}
\usepackage{soul,color}
\usepackage{amssymb,amsmath,amsthm}
\usepackage{algorithm,algorithmic}
\usepackage{cite}
\usepackage{graphicx,epstopdf}
\usepackage{caption}

\newtheorem{definition}{\textbf{Definition}}
\newtheorem{lemma}{\textbf{Lemma}}

\newtheorem{corollary}{\textbf{Corollary}}
\begin{document}

\title{Optimal Scheduling of Friendly Jammers for Securing Wireless Communication}

%\author{\IEEEauthorblockN{Jialin Wan}
%\IEEEauthorblockA{School of Computer Science and\\
%Technology, Harbin Institute of\\
%Technology, China\\
%Email: wanjialinhit@126.com}
%\and
%\IEEEauthorblockN{Siyao Cheng}
%\IEEEauthorblockA{School of Computer Science and\\
%Technology, Harbin Institute of\\
%Technology, China\\
%Email: csyhit@126.com}
%\and
%\IEEEauthorblockN{Jianzhong Li}
%\IEEEauthorblockA{School of Computer Science and\\
%Technology, Harbin Institute of\\
%Technology, China\\
%Email: lijzh@hit.edu.cn}}

\author{\IEEEauthorblockN{Jialin Wan, Siyao Cheng, Shanshan Han, Jianzhong Li}
\IEEEauthorblockA{School of Computer Science and Technology, Harbin Institute of Technology, China}
}

\maketitle

\begin{abstract}
  Wireless communication systems, such as wireless sensor networks and RFIDs, are increasingly adopted to transfer potential highly sensitive information. Since the wireless medium has a sharing nature, adversaries have a chance to eavesdrop confidential information from the communication systems. Adding artificial noises caused by friendly jammers emerges as a feasible defensive technique against adversaries. This paper studies the schedule strategies of friendly jammers, which are randomly and redundantly deployed in a circumscribed geographical area and can be unrechargeable or rechargeable, to maximize the lifetime of the jammer networks and prevent the cracking of jamming effect made by the eavesdroppers, under the constraints of geographical area, energy consumption, transmission power, and threshold level. An approximation algorithm as baseline is first proposed using the integer linear programming model. To further reduce the computational complexity, a heuristic algorithm based on the greedy strategy that less consumption leads to longer lifetime is also proposed. Finally, extensive simulation results show that the proposed algorithms are effective and efficient.
\end{abstract}

\section{Motivation}
In recent years, wireless communication technology has been widely used in our daily life. Wireless communication systems, such as wireless sensor networks and RFIDs, are increasingly adopted to transfer potential highly sensitive information \cite{DTIC2000DJF,IWWIBSN2004DTMS}. Since the wireless medium has a sharing nature, there exist great challenges in securing the sensitive information transferred by wireless communication systems, and adversaries have opportunities to attack the systems or eavesdrop confidential information from the systems.

Recently, a few researchers focus on utilizing artificial noises caused by friendly jammers to secure wireless communication systems. Karim \emph{et.al.} proposed a method for deploying minimum number of jammers powered by a central power station to secure the wireless communication systems in a circumscribed geographical area\cite{MobiHoc2012SKASVSM}. The method works well in the case that the adversaries can not enter the circumscribed geographical area but can attack or eavesdrop the wireless communication when they are near the area. To minimize power consumption of jammers, \cite{MobiHoc2012SKASVSM,MobiHoc2015YYAMEGJS} further proposed excellent methods. However, the optimal geographical layouts of jammers given in \cite{MobiHoc2012SKASVSM} are unchangeable. Thus, the jamming effect produced by the methods could be cracked by eavesdroppers with strong capability \cite{SP2013NLAS}. Moreover, this system depends heavily on a central power supply, consequently is unavailable when power failure arises. To solve these problems, we suggest an alternative method by randomly and redundantly deploying battery-powered jammers first, and then dynamically scheduling the activities of jammers to secure the wireless communication systems and prevent the cracking of jamming effect made by the eavesdroppers. The critical open problem of the alternative method is the problem of optimally scheduling the redundant jammers for securing wireless communication. Further more, rechargeable jammers can also be used to prolong the survival time of jammer networks.

\begin{figure}
\centering
\includegraphics[width=6cm]{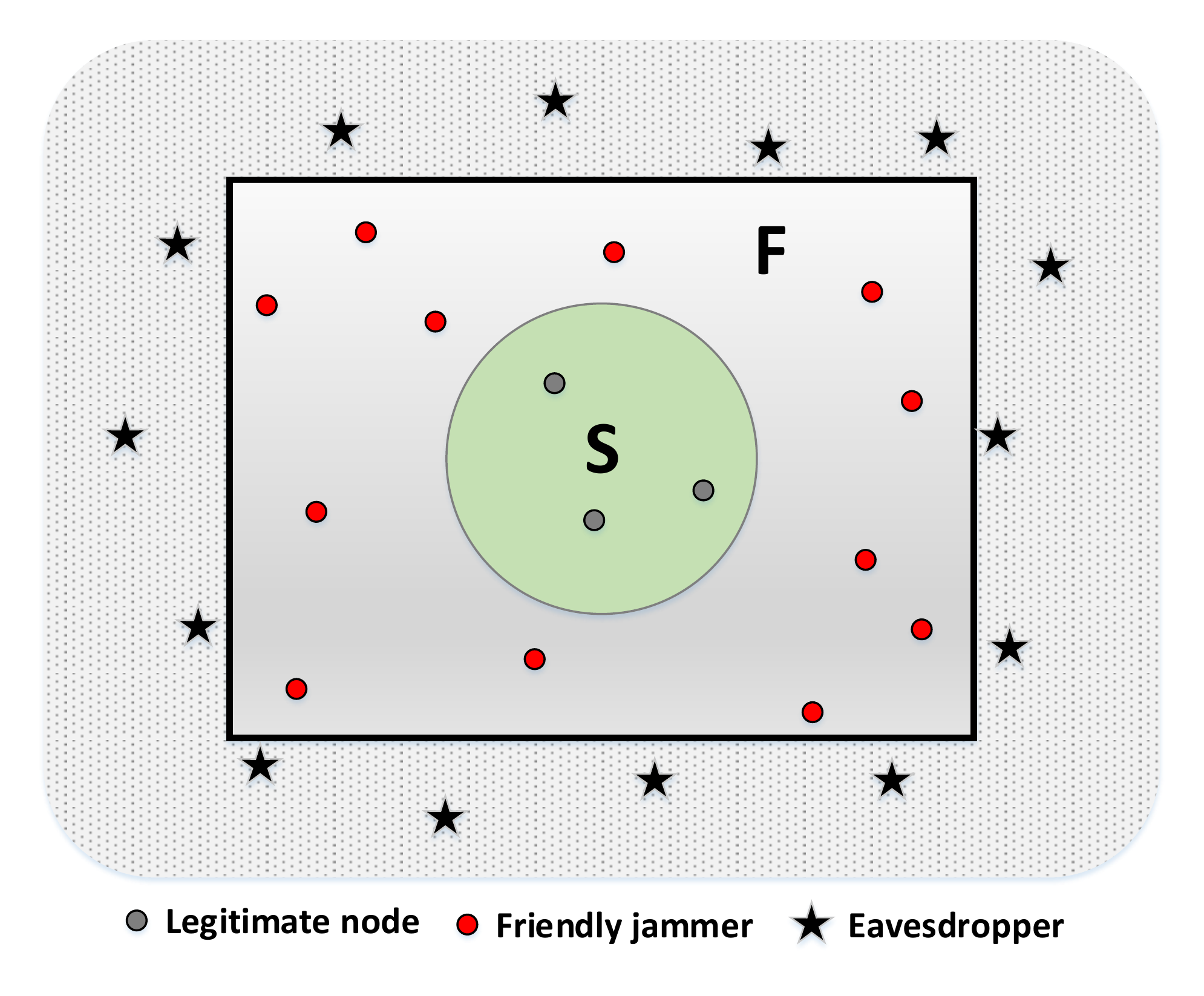}
\caption{An example Storage/Fence environment model}
\label{fig:geography}
\end{figure}

This paper focuses on the problem of optimally scheduling jammers, which are randomly and redundantly deployed and can be unrechargeable or rechargeable, for wireless communication security in a circumscribed geographical area to prevent the cracking of jamming effect made by the eavesdroppers and maximize the lifetime of the jammer networks. As shown in Figure \ref{fig:geography}, we consider the same Storage/Fence environment model as in \cite{MobiHoc2012SKASVSM}. The storage, $S$, is a closed geographic region where legitimate communication takes place. Around the storage is a fence $F$ that is used to physically prevent eavesdroppers from entering the area. Jammers, denoted by $J$, are randomly and redundantly deployed in $F \backslash S$ to cause interference. We use Signal-to-Interference-plus-Noise-Ratio (SINR), the ratio of the power of the transmitted signal to that of the interference caused by jammers, to describe the jamming effect. Based on various node characteristics, different thresholds can be determined to decide whether a receiver can hear the legitimate communication or not. Specifically, a legitimate receiver is not disturbed if the SINR of it is greater than a threshold $\delta_1$, and an eavesdropper is successfully interfered if the SINR of it is less than another threshold $\delta_2$.

The paper \cite{MobiHoc2012SKASVSM} proved that it is only needed to consider the boundary of $S$ and $F$ in the above model. Precisely speaking,
\begin{enumerate}[(i)]
\item all of the eavesdroppers located outside the fence are jammed successfully if every eavesdropper on the boundary of fence $F$ is jammed, and
\item all of the legitimate receivers inside $S$ are not jammed if any receiver on the boundary of storage $S$ is not jammed.
\end{enumerate}

With a slight abuse of notation, we refer to the boundary of storage as $S$ and the boundary of fence as $F$. Then the above description about the two threshold constraints can be formulated as
\begin{equation}\label{intro1}
\forall s \in S, SINR(s)=\frac{P_T}{\sum_{j \in J}P_J\left\|j-s\right\|^{-\gamma}} \geq \delta_1,
\end{equation}
\begin{equation}\label{intro2}
\forall p \in F, SINR(p)=\frac{P_Td(p,S)^{-\gamma}}{\sum_{j \in J}P_J\left\|j-p\right\|^{-\gamma}} \leq \delta_2.
\end{equation}
Here $P_T$ and $P_J$ are the transmission powers of legitimate nodes and jammers, respectively, $\left\|j-s\right\|$ is the Euclidean distance between $j$ and $s$, $\gamma$ is the path loss exponent, and the minimum distance between $p$ and $S$ is denoted by $d(p,S)$. Since the legitimate nodes inside $S$ are close to each other, the path loss between legitimate nodes can be ignored for brevity.

Note that both unrechargeable and rechargeable jammers are taken into account in this paper. Unrechargeable jammers have two modes, active or sleeping. When an unrechargeable jammer is active, it causes interference and consumes energy at a certain rate $c$; when in sleeping mode, it is silent and has no energy consumption. Rechargeable jammers have three modes, active, sleeping, or charging. A rechargeable jammer acts the same as unrechargeable jammer does when it is active or sleeping, and it does not cause interference and gains energy at rate $1$ when it is charging.

If we can find a series of disjoint sets of rechargeable jammers, which can be activated in a round-robin fashion to satisfy the given constraints, then the jammer network could continue working forever. Detailed analysis for this is presented in Section \uppercase\expandafter{\romannumeral6}. Unfortunately, since the budget is usually restricted in real life and the rechargeable jammers are much costlier than unrechargeable ones, the proportion of rechargeable jammers should be quite limited. So it is reasonable to assume that the total energy of the jammer network would run out eventually, even though some rechargeable jammers can be recharged. In this work, we are interested in developing feasible scheduling solutions that maximize the survival time of the jammer network.

In summary, the key contributions of this paper are listed as follows:
\begin{enumerate}[(i)]
\item We formulate the scheduling of friendly jammers into an optimization problem to maximize the survival time of the jammer network under the constraints of geographical area, energy consumption, transmission power, and threshold level.
\item We present an approximation algorithm and a greedy heuristic algorithm to solve the optimization problem.
\item We find that the integer linear programming (ILP) can be used as a step in the above two algorithms.
\item Simulation results show that the proposed algorithms are effective and efficient.
\end{enumerate}

The rest of the paper is organized as follows. Section \uppercase\expandafter{\romannumeral2} outlines the related works. Section \uppercase\expandafter{\romannumeral3} defines the optimization problem. Section \uppercase\expandafter{\romannumeral4} proves that the optimization problem is NP-hard. Our heuristic and approximate scheduling algorithms are proposed in Section \uppercase\expandafter{\romannumeral5} and analyzed in Section \uppercase\expandafter{\romannumeral6}. Experimental results are presented in Section \uppercase\expandafter{\romannumeral7}. Section \uppercase\expandafter{\romannumeral8} concludes the full paper.

\section{Related Work}
Wireless communication security has been well studied in recent years. The well-known conventional way for security is to use cryptography technique in the upper layers \cite{ICPCC2005ANHVC}. However, these approaches are impractical in many wireless communication systems, such as wireless sensor networks and RFIDs, due to the limited computing capacity of sensor nodes and the risk of secret key to be eavesdropped. Then jammers, once used by adversaries to interfere wireless communication \cite{INFOCOM2011LWW,INFOCOM2007LKP,INFOCOM2007VASD}, were introduced to prevent confidential information from being wiretapped.

There are a vast amount of extant works on friendly jammers in the literature. Vilela \emph{et.al.} showed that contention of jammers can be used for secrecy and proposed several jammer selection policies to optimize the secure throughput in the wireless networks\cite{ICC2011JPJ}. The model contains only one pair of legitimate transmitter and receiver.
Secret keys were used to control the jamming signals, such that the jamming signals are unpredictable to unauthorized devices but recoverable by authorized ones\cite{SP2013WPXH}. The method is impractical if the secret keys are eavesdropped by bad devices.
The authors of \cite{TIFS2011JMJS} focused on the design of optimal jamming configurations based on channel state information.
\cite{WoWMoM2013JJ} aimed at finding the largest number of jammers that do not cause collisions among themselves. The problem can be reduced into the maximum independent set problem. The author did not consider energy consumption.
New authentication mechanisms by integrating jamming into the communication protocol were proposed in \cite{WNS2009IPJ}.
Relay nodes were utilized to send codewords, which are independent of the source message, to confuse the eavesdroppers\cite{TIT2008LH}.
A game theoretical approach that the source pays the jammers to interfere the eavesdroppers was introduced in \cite{WCN2010HZMND}.
The authors of \cite{TSP2013AA,TWC2008SR} focused on the strategy of power allocating between transmitting data and broadcasting interference. Even though there are so many works on friendly jammers, most of them only considered the communication security between a pair of transmitter and receiver and ignored the geographical conditions. This is unrealistic in real life.

The most related works about the problem of optimally scheduling jammers, which are randomly and redundantly deployed, for wireless communication security are as follows.
Allouche \emph{et.al.} proposed a scheme to secure wireless communication through temporal jammers deployed in a restricted geographical area\cite{MobiHoc2015YYAMEGJS}. Each jammer becomes active with a certain probability on bit instants. The scheme uses bit-error probability as a measure of communication quality. The problem aims at computing a probability for each jammer to make sure that the eavesdropper's channel quality is degraded sufficiently while the legitimate communication is protected. However, the assumption of temporal jammers poses some challenging issues for hardware design.
Arkin \emph{et.al.} focused on the problem of placing a minimum number of fixed power jammers in a geographic region\cite{IPSN2015GEYAJSM}. SINR is used as the measure of communication quality. The problem is proved to be NP-hard and an ILP-based $(1+\epsilon)$-approximation scheme is proposed to solve the problem. However, the author did not consider energy consumption, not to mention the scheduling strategy of jammers to maximize lifetime.
Karim \emph{et.al.} studied strategies for allocating and managing friendly jammers in geographically restricted areas\cite{MobiHoc2012SKASVSM}. The first optimization problem, described by linear programming, aims at assigning powers to a set of fixed jammers such that the total power assigned is minimized. The second problem is to place a minimum number of jammers with the same power. The problem is modeled as ILP and a PTAS algorithm is proposed. However, since the optimal geographical layout of jammers is unchangeable, the jamming effect produced by the methods could be cracked by eavesdroppers with strong capability \cite{SP2013NLAS}. Moreover, the system depends heavily on a central power supply, and hence is unavailable when power failure arises.

The issue of extending operational time of battery-powered wireless sensor networks was investigated in \cite{WN2005CD}. The basic idea is to first organize sensors into a maximal number of disjoint set covers, and then activate these sets in turn successively to monitor all targets. A heuristic algorithm based on mixed integer programming was then proposed to compute the sets. This work is on target monitoring by sensors, not on communication protecting with friendly jammers. Moreover, the assumption that the locations of all targets are known in advance is unrealistic in protecting communication with friendly jammers since the locations of the eavesdroppers are not fixed and often unknown. Battery-powered rechargeable networks have also drawn a lot of attention among researchers \cite{INFOCOM2011CSSJ,INFOCOM2012CSSJ,INFOCOM2013GWY,INFOCOM2013FCGCH,INFOCOM2010LSK,TMC2013HCJYXS}. Although these works mainly focused on energy harvesting in sensor networks, they provided a new perspective for the security of wireless communication by rechargeable jammers.

In summary, there have been a lot of works on wireless communication security by friendly jammers, and some of them concentrate on communication security in geographically restricted areas. However, it has never been investigated and remains open problem to optimally schedule jammers to prevent the cracking of jamming effect made by the eavesdroppers and to maximize the survival time of the jammer networks in a circumscribed geographical region.

\section{Problem Definition}
The basic geographic settings in this paper are as follows. $S$ is a closed storage, which is surrounded by a fence $F$. Friendly jammers, denoted by $J$, are randomly and redundantly deployed in the region $F \backslash S$. For simplicity, we assume that all the legitimate transmitters have an identical transmission power $P_T$ and all the friendly jammers share the same jamming power $P_J$. Suppose that time is divided into identical slots, and energy is splitted into units. During each time slot a jammer consumes $c$ units of energy if active, or gains one unit of energy when charging.

Let $R_i \subseteq J$. We say the set of jammers $R_i$ is \emph{reliable}, if the following constraints are satisfied with only jammers in $R_i$ being active simultaneously.
\begin{equation}\label{prodef1}
\forall s \in S, SINR(s)=\frac{P_T}{\sum_{j \in R_i}P_J\left\|j-s\right\|^{-\gamma}} \geq \delta_1,
\end{equation}
\begin{equation}\label{prodef2}
\forall p \in F, SINR(p)=\frac{P_Td(p,S)^{-\gamma}}{\sum_{j \in R_i}P_J\left\|j-p\right\|^{-\gamma}} \leq \delta_2.
\end{equation}
The family of \emph{reliable} sets of jammers is denoted by $R=\{R_i\mid \mbox{$R_i \subseteq J$ and $R_i$ is \emph{reliable}}\}$. A jamming schedule can be defined as a sequence of \emph{reliable} jammer sets $D_1,D_2,\cdots,D_l$, where $D_i \in R$ and $D_i$ is activated in the $i^{th}$ time slot. It is worth to note that, a jammer could be covered by several different \emph{reliable} sets, but each set is activated for only one slot. For instance, as shown in Figure \ref{fig:schedule}, all jammers in $D_1$ are activated in the first slot and each jammer consumes $c$ units of energy; then $D_1$ stops working and $D_2$ is activated in the second slot, and so on.

\begin{figure}
\centering
\includegraphics[width=8cm]{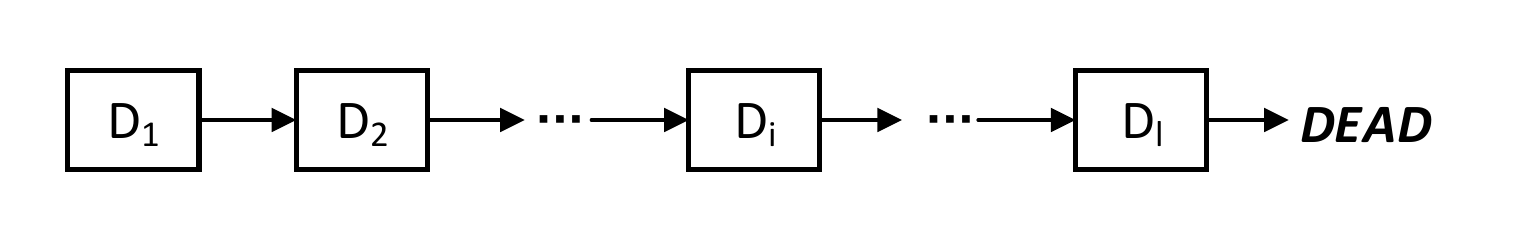}
\caption{An example jamming schedule}
\label{fig:schedule}
\end{figure}

With jammers dying gradually due to energy consumption, let's consider a situation where no such \emph{reliable} subset of $J$ can be found after the $l^{th}$ slot, and we refer to this case as \emph{DEAD}. The lifetime of the whole jammer network is defined as the maximum number of time slots before \emph{DEAD}. Various schedules result in different lifetimes of a jammer network. By appropriately scheduling, rechargeable jammers can be recharged in turn before running out of energy, and hence prolongs the lifetime of the whole jammer network. The problem of optimally scheduling jammers is to find a jamming schedule $D_1,D_2,\cdots,D_l$ such that the lifetime of a given jammer network is extended as much as possible, hence $l$ is maximized.

In the next section, we prove that the problem is NP-hard. In Section \uppercase\expandafter{\romannumeral5}, first we propose an approximation algorithm as baseline, and then a greedy heuristic algorithm based on ILP is presented to solve the problem of optimally scheduling jammers.

\section{Hardness of the Problem}
In our optimal jammers scheduling problem, we need to select a subset of jammers from $J$ in each time slot to ensure the constraints (\ref{prodef1}) and (\ref{prodef2}) are satisfied. In the most general case, both unrechargeable jammers and rechargeable jammers constitute the whole jammer set $J$ together, and each jammer has a certain life span. Let's consider a special case derived from the above general case, i.e., all of the jammers in $J$ are unrechargeable and all jammers share the identical life span $1$. That means, each jammer could be chosen to be active for only one time slot and then would die. Then we have a special problem with only unrechargeable jammers randomly distributed in the region $F \backslash S$ to cause interference to secure the legitimate communication inside the storage $S$.

Since the jammer set $J$ is limited, we can find out all of the \emph{reliable} subsets of $J$, denoted by $R_i$ $(i=1,2,\cdots$). Let $R= \{R_1,R_2,\cdots,R_m\}$ be the set of those \emph{reliable} subsets. Apparently, $R_i \subset J$. Then selecting a subset of jammers to satisfy (\ref{prodef1}) and (\ref{prodef2}) is equivalent to choosing a set from $R$. Since the life span of each jammer is $1$, a jammer could be chosen to work only once. Thus the special problem is easily converted to the classical maximum independent set problem.

Now we construct a graph $G=(V,E)$ based on $R$ as follows:
\begin{enumerate}[(i)]
\item add each $R_i \in R$ to $V$ as a vertex, and
\item $\forall R_i,R_j \in R$, if $R_i \cap R_j \neq \phi$, then add $(R_i,R_j)$ to $E$ as an edge.
\end{enumerate}
Remember that $\forall x \in J$, $x$ could be chosen only once. In other words, at most one of any two adjacent vertexes in $G$ could be chosen. The original optimal goal is to maximize the number of chosen subsets of $J$, which is equivalent to choosing the most adjacent vertexes in graph $G$. Thus, the special problem is reduced to the maximum independent set problem whose complexity is NP-Complete. Thus the complexity of the original problem is at least NP-Complete.

However, when rechargeable jammers and multiple life span of jammers are considered, the problem becomes much more complicated. Based on the analysis above, the computational complexity of our original problem of optimally scheduling jammers is NP-hard.

\section{Algorithm}
In this section, we propose an approximation algorithm, called minimum reliable set algorithm ($MRS$ for short), as baseline, and a greedy heuristic algorithm to solve the problem of optimally scheduling jammers. In the $MRS$ algorithm, we first compute a family of special \emph{reliable} subsets of jammers, called \emph{minimum reliable set}. Then the problem can be modeled as an integer linear programming problem based on the calculated \emph{minimum reliable set}, which can be solved by an existing polynomial time algorithm with a certain approximation ratio. Note that the $MRS$ algorithm is specially designed for situations where there are only unrechargeable jammers in the network. The heuristic algorithm is based on an intuitive greedy strategy that the less energy is consumed in each time slot, the longer the lifetime of the jammer network will be. The algorithm greedily computes a \emph{reliable} subset of jammers, $D_i \in R$, with the minimum energy decrease in each iteration. This process will not stop until the remaining jammers are not enough to satisfy the given threshold constraints. The final output of the greedy algorithm is the sequence $D_1,D_2,\cdots,D_l$, which means that the algorithm terminates after $l$ iterations and the lifetime of the whole jammer network is $l$.

\subsection{MRS Algorithm}
In this subsection, we focus on the special case with only unrechargeable jammers in the jammer network. We first discuss the definition of the \emph{minimum reliable solution}, then give the proof that it is sufficient to make choices just from the \emph{minimum reliable set}, and finally present the $MRS$ algorithm and prove its correctness.

\begin{definition}
For any $M_i \subseteq J$, $M_i$ is called a \emph{minimum reliable solution} if the following two conditions are satisfied:
\begin{enumerate}[(i)]
\item $M_i \in R$, which means $M_i$ meets the threshold constraints (\ref{prodef1}) and (\ref{prodef2});
\item $\forall N \subset M_i, N \notin R$, which means $N$ cannot meet the threshold constraints (\ref{prodef1}) and (\ref{prodef2}).
\end{enumerate}
The family of all the \emph{minimum reliable solutions} is called the \emph{minimum reliable set}, denoted by $M$.
\begin{equation*}
M=\{M_i\mid \mbox{$M_i \in R$ and $\forall N \subset M_i,N \notin R$}\}.
\end{equation*}
\end{definition}

\begin{lemma}\label{lemma1}
If a jammer is replaced by a new jammer with more battery power, then the lifetime of the whole jammer network will increase, or at least remain the same.
\end{lemma}

\begin{lemma}\label{lemma2}
Consider an optimal schedule $D_1,D_2,\cdots,D_l$. Suppose that one of the \emph{reliable} set, without loss of generality, $D_i$, is not a \emph{minimum reliable solution}. Then we can construct a new optimal schedule by replacing $D_i$ with the corresponding \emph{minimum reliable solution} $D_i^{'}$.
\end{lemma}

\begin{proof}
Since $D_i^{'}$ is the \emph{minimum reliable solution} corresponding to $D_i$, then $D_i^{'} \subset D_i$.
Let's consider the schedule $D_1,\cdots,D_{i-1},D_i^{'},\cdots$, as shown in Figure \ref{fig:replace}. After the $i^{th}$ time slot, the remaining battery power of the whole network is more than that under the original schedule. According to Lemma \ref{lemma1}, the remaining lifetime of the network is no less than $l-i$, thus the total lifetime is at least $l$. Since $D_1,D_2,\cdots,D_l$ is an optimal schedule, $D_1,\cdots,D_{i-1},D_i^{'},\cdots$ is optimal, too.
\end{proof}

\begin{figure}
\centering
\includegraphics[width=8cm]{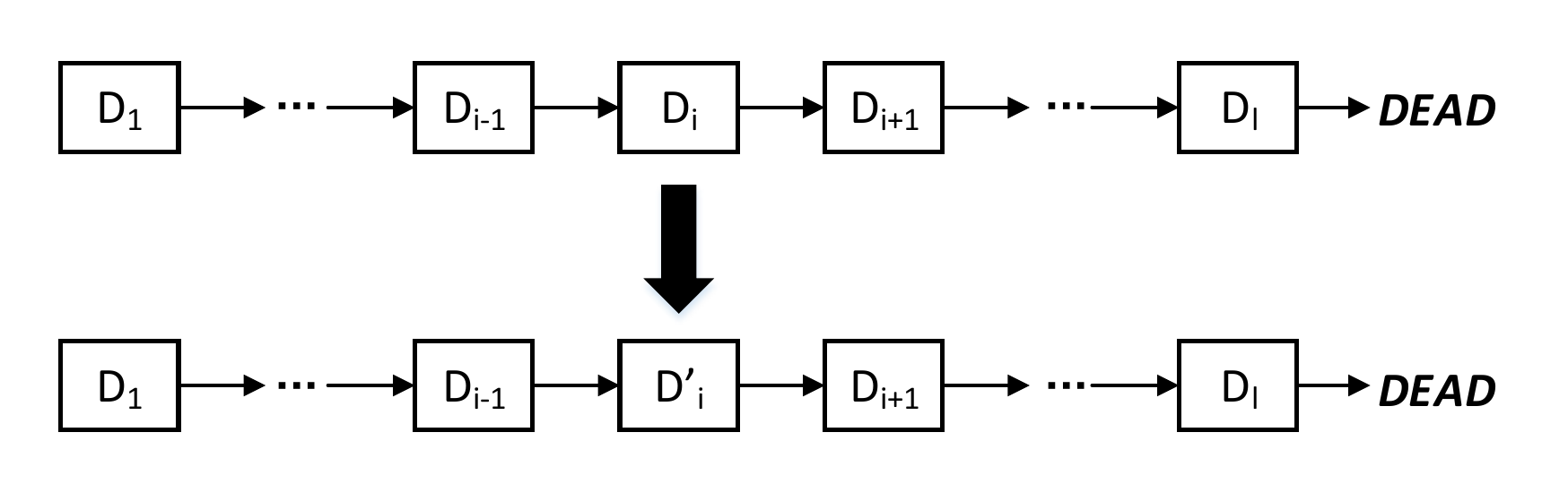}
\caption{Construct an optimal schedule by replacing $D_i$ with $D_i^{'}$}
\label{fig:replace}
\end{figure}

Based on what have been discussed above, when given an optimal schedule, we can always find a new optimal schedule, which includes only \emph{minimum reliable solutions}, by replacing all the \emph{reliable} sets with corresponding \emph{minimum reliable solutions}. Then we have the following corollary.

\begin{corollary}
We can always find an optimal schedule $D_1,D_2,\cdots,D_l$, such that for all $D_i$, $D_i$ is a \emph{minimum reliable solution}, that is, $D_i \in M$.
\end{corollary}

Now we discuss the idea of $MRS$ algorithm. Based on the above corollary, only the \emph{minimum reliable set} should be considered to maximize the lifetime of the whole jammer network. In the first phase of $MRS$ algorithm, we need to compute the \emph{minimum reliable set} pertaining to the given $S,F$ and $J$. In the second phase, we should make a schedule from that \emph{minimum reliable set}, to determine which \emph{minimum reliable solution} should be active during each time instant.

\begin{algorithm}[h]
    \caption{MRS Algorithm}
    \begin{algorithmic}[1]
    \REQUIRE~~\\
        The set of jammers $J$;\\
        The boundary of storage $S$ and fence $F$;\\
        Parameters: $P_T$, $P_J$, $\delta_1$, $\delta_2$, $B$, $c$.
    \ENSURE~~\\
        The final jamming schedule.
    \STATE compute minimum reliable set $M$;\\
    \STATE construct the ILP model based on constraints of life span;\\
    \STATE compute the ILP problem approximately;
    \end{algorithmic}
\end{algorithm}

Let $J= \{j_1,j_2,\cdots,j_n\}$ denotes the set of jammers, and $W= \{w_1,w_2,\cdots,w_n\}$ is the corresponding life span of each jammer. Suppose that $M= \{M_1,M_2,\cdots,M_m\}$ is the computed \emph{minimum reliable set}, and
$Res= \langle n_1,n_2,\cdots,n_m\rangle$ is the optimal solution of the \emph{minimum reliable set}, which means the output schedule of jammers is the following sequence, $n_1M_1,n_2M_2,\cdots,n_mM_m$. Then the total lifetime is $\sum_{k=1}^{m}n_k$. The following indicator function $I_{M_k}(j_i)$ shows whether jammer $j_i$ is contained in $M_k$ or not.
\begin{equation*}
\forall j_i \in J, I_{M_k}(j_i)=
\begin{cases}
1, &\mbox{$j_i$ is in $M_k$};\\
0, &\mbox{$j_i$ is not in $M_k$}.
\end{cases}
\end{equation*}

Once a jammer runs out of battery power, it would die and can not be active anymore. Since all of the jammers are unrechargeable and the battery power could not increase, each jammer has a limited life span. Thus the following constraints must be satisfied.
\begin{equation*}
\forall j_i \in J,
\sum_{k=1}^{m}I_{M_k}(j_i)n_k \leq w_i.
\end{equation*}
The lifetime of the network can be formulated as $l=\sum_{k=1}^{m}n_k$.
%\begin{equation*}
%l=\sum_{k=1}^{m}n_k.
%\end{equation*}
Then the optimization goal of the problem is to maximize the lifetime $l$. This is a typical integer linear programming problem and can be solved approximately in polynomial time.

For example, assume that $J$=$\{a,b,c,d,e,f\}$, $W$=$\{2,2,2,2,$
$2,2\}$, the \emph{minimum reliable set} $M=\{M_1,M_2,M_3\}$, and the \emph{minimum reliable solutions} $M_1=\{a,b,c,e\}$, $M_2=\{a,c,$
$d\}$, $M_3= \{b,f\}$. Then the life span constraints can be shown as follow:
\begin{equation*}
\begin{cases}
n_1+n_2 \leq w_a=2,\\
n_1+n_3 \leq w_b=2,\\
n_1+n_2 \leq w_c=2,\\
n_2 \leq w_d=2,\\
n_1 \leq w_e=2,\\
n_3 \leq w_f=2.\\
\end{cases}
\end{equation*}
The goal is to maximize
\begin{equation*}
l=n_1+n_2+n_3.
\end{equation*}

Under the settings mentioned above, if we make choices from the given \emph{minimum reliable set}, $Res= \langle 1,1,1\rangle$ is a feasible solution with lifetime $l=3$, while $Res^*= \langle 0,2,2\rangle$ is the optimal solution and the corresponding maximum lifetime is $l^*=4$. However, if we do not make any schedule and rudely activate all jammers in each time slot, then the lifetime of the whole jammer network is just $l=2$.

\subsection{Greedy Algorithm}
In this subsection, we focus on the most general cases with both unrechargeable jammers and rechargeable ones in the network. Given storage $S$, fence $F$ and the set of jammers $J$, we hope to determine a jamming schedule $D_1,D_2,\cdots,D_l$ such that $l$ is maximized. As mentioned above, the jamming network is a hybrid structure of unrechargeable jammers and rechargeable ones. The unrechargeable jammers have two different modes, active and sleeping. An unrechargeable jammer consumes $c$ units of energy when active and maintains a constant energy when sleeping. Similarly, the rechargeable jammers have three modes, active, sleeping, and charging. A rechargeable jammer acts the same as unrechargeable jammer does when active or sleeping, and gains one unit of energy when charging. It is worth to note that, a rechargeable jammer is in charging mode by default when not active, and is sleeping only when full charged and not active. The greedy strategy is to minimize the amount of consumed energy units minus the amount of increased energy units, denoted by $\Delta$, under the given threshold constraints in each time slot. Each iteration can be described by an ILP model.

We denote the set of unrechargeable jammers as $J_U \subset J$ and the set of rechargeable jammers as $J_R \subset J$. Apparently, $J_U \cup J_R = J$. According to the remaining battery power of each jammer, the set of jammers $J$ can be divided into three groups: dead jammers $J_{d}$, full charged jammers $J_{f}$ and other normal jammers $J_{n}$. We use binary variables $c_i$ for each jammer $i \in J$ to indicate whether $i$ is chosen to be active or not in a time instant. Note that only chosen jammers with $c_i=1$ are activated to secure legitimate communication inside the storage. Thus the threshold constraints (\ref{prodef1}) and (\ref{prodef2}) should be modified with $c_i$ taken into account. The modified constraints are listed as follows:
\begin{equation*}
\forall s \in S, \frac{P_T}{\sum_{i \in J}c_iP_J\left\|i-s\right\|^{-\gamma}} \geq \delta_1,
\end{equation*}
\begin{equation*}
\forall p \in F, \frac{P_Td(p,S)^{-\gamma}}{\sum_{i \in J}c_iP_J\left\|i-p\right\|^{-\gamma}} \leq \delta_2.
\end{equation*}
That is to say,
\begin{equation}\label{finalconstr1}
\forall s \in S, \sum_{i \in J}c_i\left\|i-s\right\|^{-\gamma} \leq \frac{P_T}{P_J\delta_1},
\end{equation}
\begin{equation}\label{finalconstr2}
\forall p \in F, \sum_{i \in J}c_i\left\|i-p\right\|^{-\gamma} \geq \frac{P_Td(p,S)^{-\gamma}}{P_J\delta_2}.
\end{equation}

According to the proportion of rechargeable jammers in $J$, we can divide the problem into three different situations, i.e., pure unrechargeable jammers(UJs), pure rechargeable jammers(RJs), and the hybrid of UJs and RJs.

\subsubsection{Pure UJs}
Since there are no rechargeable jammers in the jammer network, the increased energy must be zero. Then the greedy strategy is to minimize
\begin{equation*}
\Delta=\sum\nolimits_{i \in J_{f} \cup J_{n}}c \cdot c_i=c\sum\nolimits_{i \in J_{f} \cup J_{n}}c_i.
\end{equation*}
\subsubsection{Pure RJs}
Remember that rechargeable jammers can be active, charging or sleeping. Then the greedy strategy is to minimize
\begin{equation*}
\begin{aligned}
\Delta&=\sum\nolimits_{i \in J_{n}}\left[c \cdot c_i-\left(1-c_i\right)\right]+\sum\nolimits_{i \in J_{f}}c \cdot c_i \\
&=\left(c+1\right)\sum\nolimits_{i \in J_{n}}c_i+c\sum\nolimits_{i \in J_{f}}c_i-\left|J_{n}\right|.
\end{aligned}
\end{equation*}
\subsubsection{The hybrid of UJs and RJs}
For all jammers in $J_U$, the total energy consumption is
\begin{equation*}
\begin{aligned}
\Delta_U=\sum\nolimits_{i \in J_U \cap(J_{f} \cup J_{n})}c \cdot c_i.
\end{aligned}
\end{equation*}
For other jammers in $J_R$, the energy decrease is
\begin{equation*}
\begin{aligned}
\Delta_R=\sum\nolimits_{i \in J_R \cap J_{n}}\left[c \cdot c_i-\left(1-c_i\right)\right]+\sum\nolimits_{i \in J_R \cap J_{f}}c \cdot c_i.
\end{aligned}
\end{equation*}
Then the greedy strategy is to minimize
\begin{equation*}
\begin{aligned}
\Delta=&\Delta_U+\Delta_R\\
=&c\sum\nolimits_{i \in J_U \cap(J_{f} \cup J_{n})}c_i+\left(c+1\right)\sum\nolimits_{i \in J_R \cap J_{n}}c_i\\
+&c\sum\nolimits_{i \in J_R \cap J_{f}}c_i-\left|J_R \cap J_{n}\right|.
\end{aligned}
\end{equation*}

\begin{algorithm}[h]
    \caption{Greedy Algorithm}
    \begin{algorithmic}[1]
    \REQUIRE~~\\
        The set of jammers $J$;\\
        The remaining battery of each jammer, $B[1:\left|J\right|]$;\\
        Parameters: $P_T$, $P_J$, $\delta_1$, $\delta_2$.
    \ENSURE~~\\
        The final jamming schedule, $D$.
    \STATE $D \leftarrow \phi$;
    \STATE $j \leftarrow 0$;
    \WHILE{(true)}
        \STATE Initial $c[1:\left|J\right|]$;\\
        \STATE Compute $ILP\_result$ such that $\Delta$ is minimized;\\
        \IF{($ILP\_result \neq NULL$)}
            \STATE Update $B[1:\left|J\right|]$;\\
            \STATE $j \leftarrow j+1$;\\
            \STATE $D[j] \leftarrow \phi$;\\
            \FOR{$i \leftarrow 1$ to $\left|J\right|$}
                \IF{($c[i]=1$)}
                    \STATE $D[j] \leftarrow D[j] \cup \{i\}$;\\
                \ENDIF
            \ENDFOR
            \STATE $D \leftarrow D \cup \{D[j]\}$;\\
        \ELSE
            \RETURN $D$;
        \ENDIF
    \ENDWHILE
    \end{algorithmic}
\end{algorithm}

As discussed above, we have three different optimization goals pertaining to different situations in each iteration. What calls for special attention is that all of the three optimization functions are linear, and the constraints (\ref{finalconstr1}) and (\ref{finalconstr2}) are linear inequalities, too. So the optimization problem in each iteration becomes a typical ILP problem and can be solved approximately in polynomial time with a certain approximation ratio. In each time slot, the algorithm computes the near optimal solution of the above ILP problem, denoted by $D_i$, which forms the final sequence of jamming schedule. Suppose the total number of iterations is $l$. On the termination of the algorithm, the sequence $D_1,D_2,\cdots,D_l$ is the ultimate output and the corresponding lifetime is $l$.

\section{Performance Analysis}
\subsection{Complexity of Greedy Algorithm}
The complexity of our greedy algorithm is dominated by the complexity of approximately solving the ILP problem and the number of iterations. Firstly, to approximately solve an ILP problem, a famous approach is to compute the corresponding LP problem, and then round the solution of LP as the final solution of the ILP problem. According to \cite{MP1991Y}, the best performance of the LP solver is $O(n^3)$ using Ye's algorithm, where $n$ is the number of variables. Secondly, note that the total battery power of all the jammers is $nB$, where $n$ is the number of jammers and $B$ is the initial battery power of each jammer. In each time slot, the whole jammer network will consume at least one unit of energy (otherwise the network could continue working forever), thus the number of iterations, or lifetime, is no more than $nB$. Hence, the complexity of our greedy algorithm is $O(Bn^4)$.
\subsection{Pruning Strategy}
In this subsection, we present a naive method to estimate the range of the number of active jammers in each time instant. The estimation result can be utilized to reduce the operation steps in our algorithms. The method gives the upper bound of the number of active jammers by computing the range of distance from jammers to points $s$ on $S$, and shows the lower bound by determining the range of distance from jammers to points $p$ on $F$.

For any node $s$ on the boundary of $S$, the distance from a jammer $j$ to it, denoted by $\left\|j-s\right\|$, meets the following inequality:
\begin{equation*}
\min_{j \in J}\left\|j-s\right\| \leq \left\|j-s\right\| \leq \max_{j \in J}\left\|j-s\right\|,
\end{equation*}
then we can deduce from Formula (\ref{finalconstr1}) the following statement:
\begin{equation*}
(\max_{j \in J}\left\|j-s\right\|)^{-\gamma}\sum_{j \in J}c_j \leq \sum_{j \in J}c_j\left\|j-s\right\|^{-\gamma} \leq \frac{P_T}{P_J\delta_1},
\end{equation*}
then,
\begin{equation*}
\sum_{j \in J}c_j \leq \frac{P_T}{P_J\delta_1}(\max_{j \in J}\left\|j-s\right\|)^{\gamma}, \forall s \in S,
\end{equation*}
thus, we get the upper bound of the number of active jammers
\begin{equation*}
\sum_{j \in J}c_j \leq \min_{s \in S}{\frac{P_T}{P_J\delta_1}(\max_{j \in J}\left\|j-s\right\|)^{\gamma}}.
\end{equation*}
Analogously, we have the following inequality deduced from Formula (\ref{finalconstr2}) that for any node $p$ on the boundary of $F$,
\begin{equation*}
(\min_{j \in J}\left\|j-p\right\|)^{-\gamma}\sum_{j \in J}c_j \geq \sum_{j \in J}c_j\left\|j-p\right\|^{-\gamma} \geq \frac{P_Td(p,S)^{-\gamma}}{P_J\delta_2},
\end{equation*}
then,
\begin{equation*}
\sum_{j \in J}c_j \geq \frac{P_Td(p,S)^{-\gamma}}{P_J\delta_2}(\min_{j \in J}\left\|j-p\right\|)^{\gamma}, \forall p \in F,
\end{equation*}
consequently we get the lower bound
\begin{equation*}
\sum_{j \in J}c_j \geq \max_{p \in F}{\frac{P_Td(p,S)^{-\gamma}}{P_J\delta_2}(\min_{j \in J}\left\|j-p\right\|)^{\gamma}}.
\end{equation*}

Based on the analysis above, the number of active jammers during each time slot should be bounded within the estimated upper and lower bound. In other words, any jammer set with less or more jammers would not meet the constraints (\ref{prodef1}) and (\ref{prodef2}) and should be excluded. Thus the steps needed in the algorithms could be reduced.

\subsection{Analysis of Lifetime}
In this subsection, we discuss the issue whether the jammer network can continue working forever or not, or in other words, whether the lifetime can be prolonged indefinitely.

Since unrechargeable jammers can only consume energy but can not be charged, the number of time slots during which they are active is limited by the initial battery power $B$ and the rate of energy consumption $c$. Specifically, an unrechargeable jammer can be chosen to cause interference for at most $\left\lfloor \frac{B}{c}\right\rfloor$ times and would die then. Therefore, whether an unlimited lifetime can be reached or not mainly depends on the number of rechargeable jammers.

\begin{figure}
\centering
\includegraphics[width=5cm]{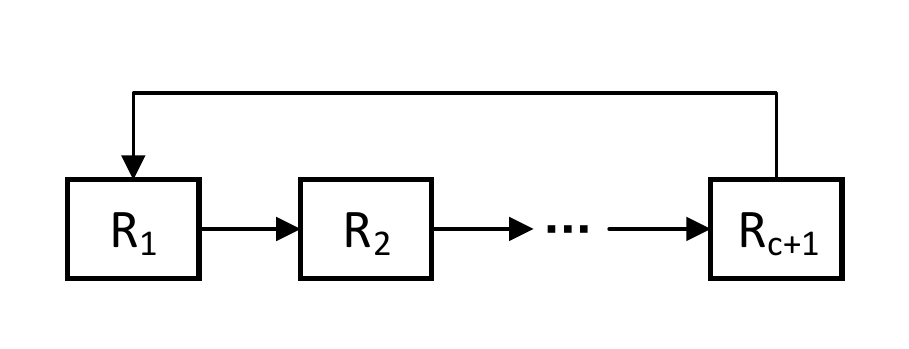}
\caption{c+1 disjoint reliable sets works in a round-robin fashion}
\label{fig:round-robin}
\end{figure}

Analogous to the \emph{Maximum Disjoint Set Covers} problem studied in \cite{ICC2001SP,WN2005CD}, if we find $c+1$ disjoint \emph{reliable} subsets of $J_R$, denoted by $\{R_1,R_2,\cdots,R_{c+1}\}$, then these subsets could be activated in turn, in a round-robin fashion, such that during each time slot only one subset is responsible for securing the communication, while all other jammers are in sleeping or charging mode. As shown in Figure \ref{fig:round-robin}, consider a period of $c+1$ successive time slots, each subset of jammers would be active for only one slot and charging for $c$ slots. From the global perspective, the remaining battery power of each jammer would remain constant after a period. That is to say, the whole jammer network does not suffer from energy decrease and could keep working successfully forever. However, if no such $c+1$ disjoint \emph{reliable} subsets are found, then each jammer does not have enough time to refill the already consumed $c$ units of energy, thus the remaining battery power of each jammer would decrease and finally run out. As a result of that, the jammer network would die eventually.

Given the storage $S$, fence $F$ and the randomly deployed jammers $J$, we can determine the least number of active jammers , $L_{jam}$, to satisfy the threshold constraints, just as mentioned above. If the total number of rechargeable jammers is less than $(c+1)L_{jam}$, then it is obviously impossible to find such $\{R_1,R_2,\cdots,R_{c+1}\}$, hence we can assert that the jammer network would die eventually; otherwise, it is possible to make an unlimited schedule for the network. Unfortunately, since the budget is restricted and the rechargeable jammers are much costlier than unrechargeable ones, the proportion of rechargeable jammers should be quite limited. So it is reasonable to assume that the network lifetime is limited in our work. Future work will focus on the relationship between unlimited lifetime and the proportion of rechargeable jammers.

\section{Simulation Results}

\subsection{Simulation Setup}
\begin{figure}
\centering
\includegraphics[width=4cm]{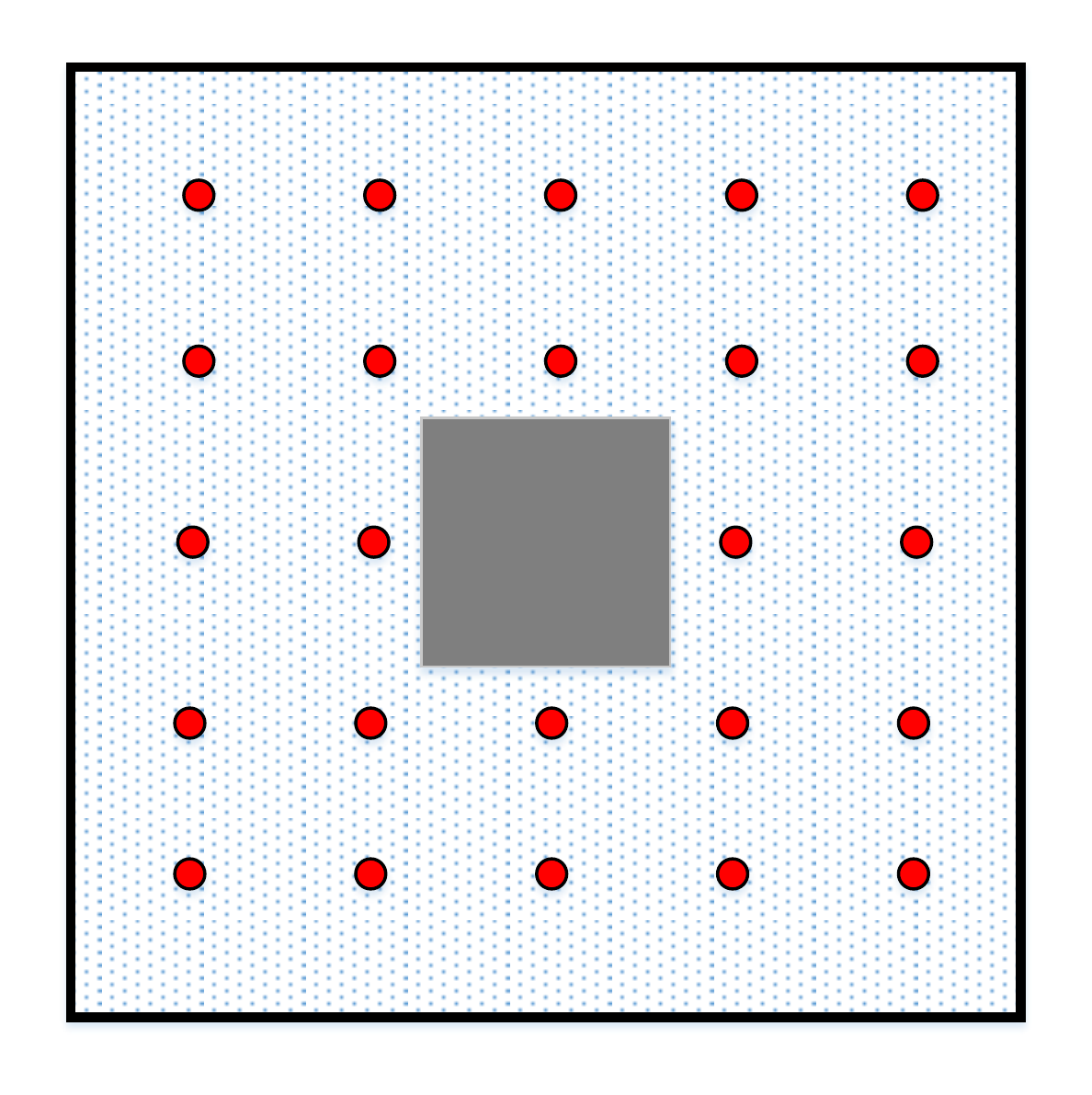}
\caption{Storage(grey shaded area)/Fence(black line) model with randomly deployed jammers(red dots)}
\label{fig:simulation}
\end{figure}

In this section, we conduct extensive simulations to evaluate the performance of our greedy algorithm under different network settings. The basic geographical setting we have chosen is shown in Figure \ref{fig:simulation}. The fence $F$ is a $100m \times 100m$ square area, and the storage $S$ is a $25m \times 25m$ area located right in the middle of $F$. According to the discretization strategy of both $F$ and $S$ in \cite{MobiHoc2012SKASVSM}, we set the step size $\lambda$ as $2m$, then the boundary of $S$ and $F$ can be divided into a discrete set of spots. We simulate a jammer network with both unrechargeable jammers and rechargeable jammers randomly located in $F \backslash S$. In this simulation, we consider the following tunable parameters:
\begin{itemize}
\item $n$, the number of jammers. We vary the number of randomly deployed jammers between $30$ and $120$ to study the relation between network lifetime and density of jammers. The default value is $n=100$.
\item $P_J$, the jamming power of all jammers. We vary $P_J$ between $0.1$ and $10$ to study the effect of jamming power on the network lifetime. The default value is $P_J=1$.
\item $B$, the life span of each jammer. It varies between $1$ and $10$. The default value is $B=10$.
\item $\delta_2$, the threshold level indicating how much capable the eavesdropper is. We vary $\delta_2$ between $0.1$ and $0.9$. The default value is $\delta_2=0.5$.
\item $\eta$, the percentage of rechargeable jammers in the whole set of jammers. It varies between $0$ and $0.8$. The default value is $\eta=0$.
\item $c$, the rate of energy consumption of each jammer. It varies between $4$ and $20$. The default value is $c=10$.
\end{itemize}

We implement our greedy algorithm with C++ and use the optimization toolbox in LINGO to solve the linear programming. For each setting of the parameters, we repeat the experiment $5$ times for different jammers' random deployment. The base value of lifetime under the hypothesis that all jammers are active in each time slot is used as reference.

\begin{figure*}
\begin{minipage}[t]{0.33\linewidth}
\centering
\includegraphics[width=5cm]{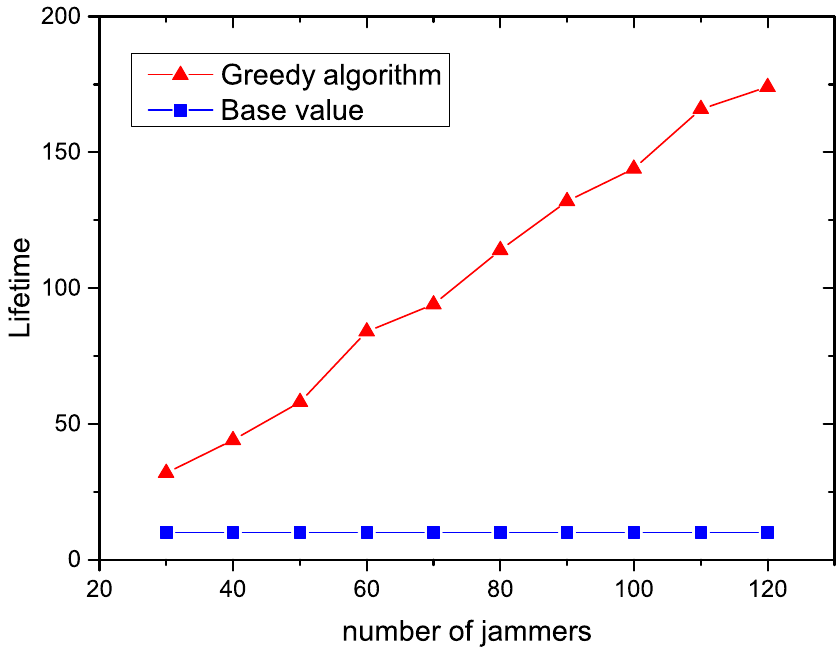}
\caption{Lifetime Vs number of jammers}
\label{fig:jammer}
\end{minipage}%
\begin{minipage}[t]{0.33\linewidth}
\centering
\includegraphics[width=5cm]{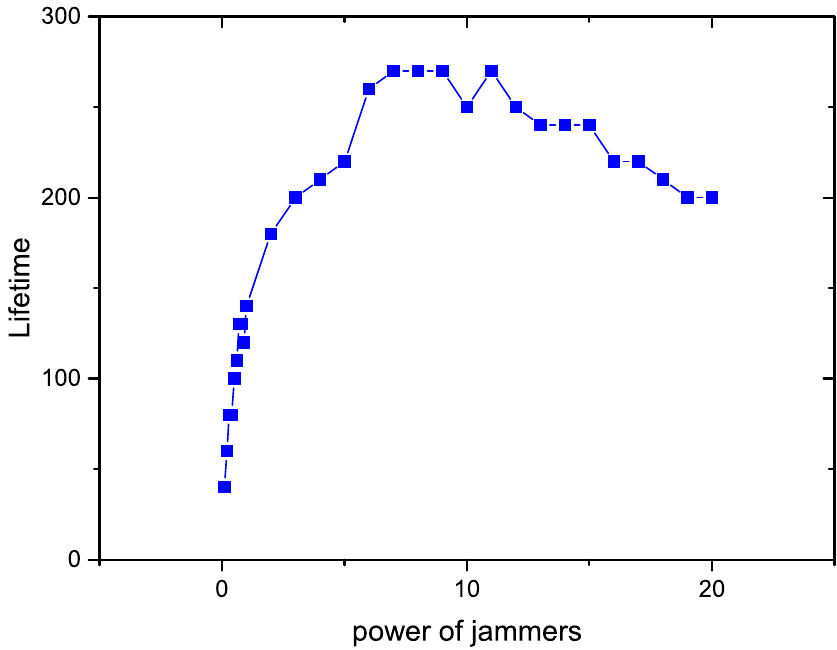}
\caption{Lifetime Vs power of jammers}
\label{fig:power}
\end{minipage}%
\begin{minipage}[t]{0.33\linewidth}
\centering
\includegraphics[width=5cm]{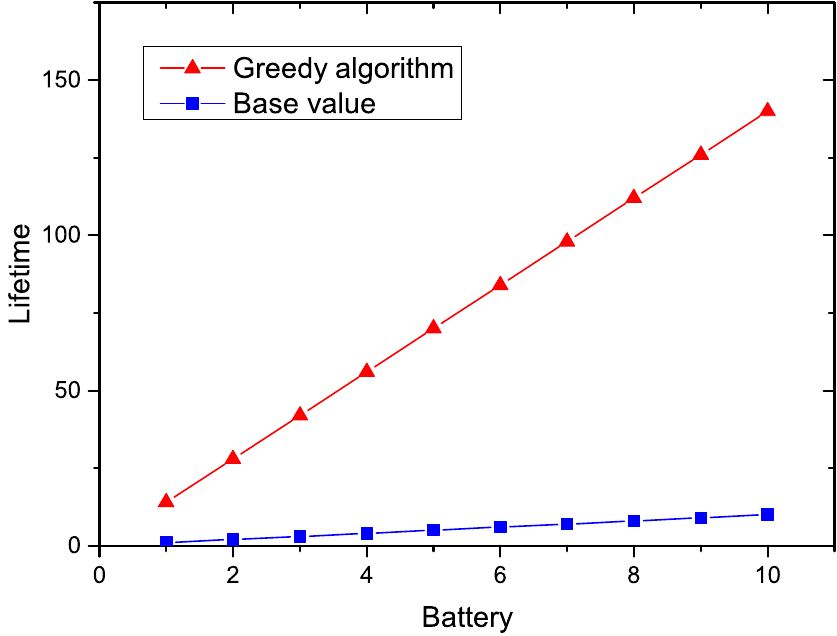}
\caption{Lifetime Vs battery of jammers}
\label{fig:battery}
\end{minipage}
\end{figure*}

\begin{figure*}
\begin{minipage}[t]{0.33\linewidth}
\centering
\includegraphics[width=5cm]{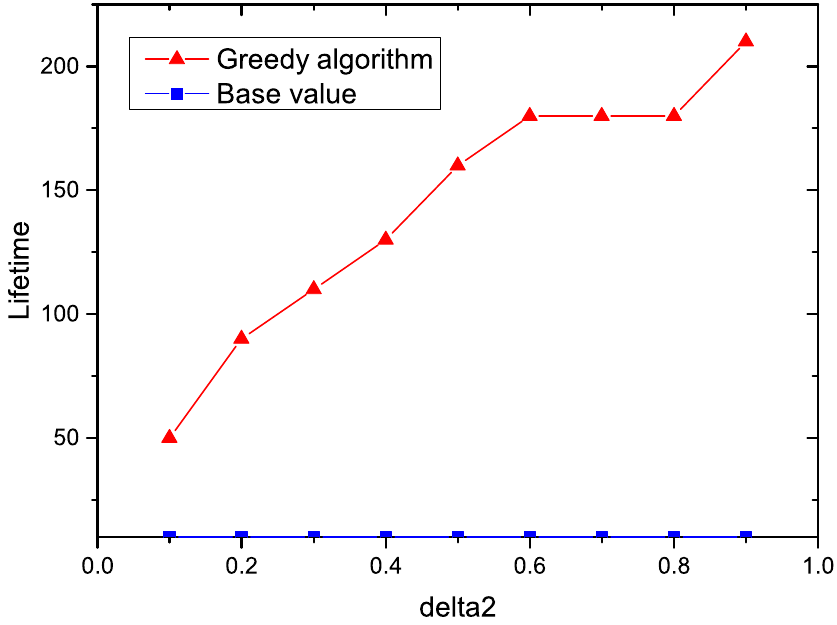}
\caption{Lifetime Vs threshold level}
\label{fig:delta}
\end{minipage}%
\begin{minipage}[t]{0.33\linewidth}
\centering
\includegraphics[width=5cm]{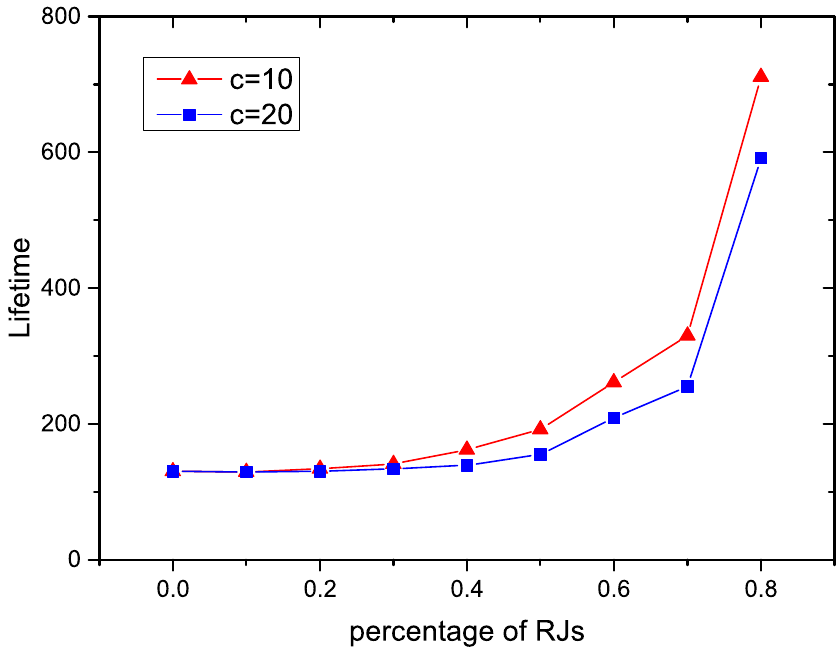}
\caption{Lifetime Vs percentage of RJs}
\label{fig:sigma}
\end{minipage}%
\begin{minipage}[t]{0.33\linewidth}
\centering
\includegraphics[width=5cm]{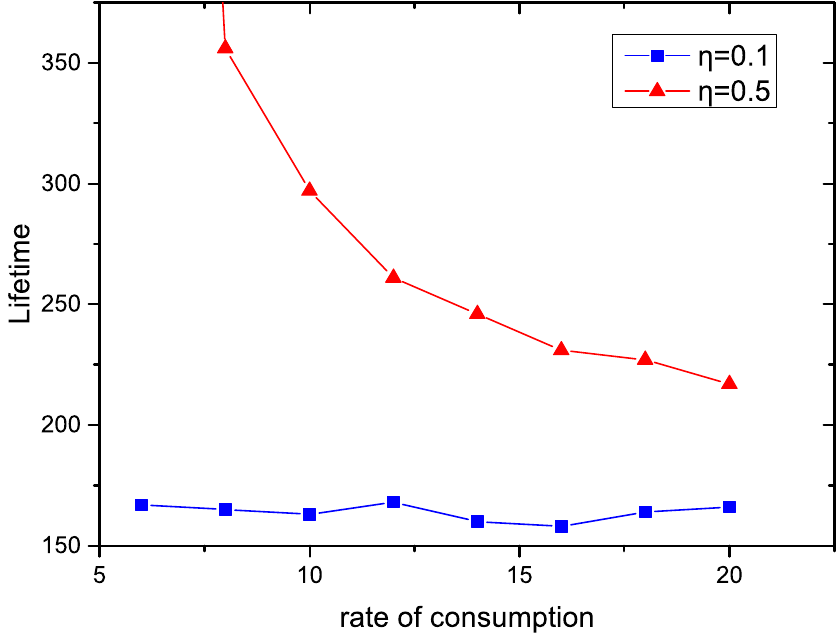}
\caption{Lifetime Vs rate of consumption}
\label{fig:consume}
\end{minipage}
\end{figure*}

\subsection{Analysis of Simulation Results}
In Figure \ref{fig:jammer}, we present the network lifetime computed by the greedy algorithm, depending on the number of randomly distributed jammers $n$, which varies between $30$ and $120$ with an increment of $10$. Network lifetime returned by the greedy algorithm increases nearly linearly with the increase of jammer density and is far greater than the base value $10$. When more jammers are deployed, a jammer in a \emph{reliable} subset can be replaced by more alternatives, then more such \emph{reliable} subsets can be found, consequently a longer lifetime can be achieved.

In Figure \ref{fig:power}, we study the impact of the jamming power on network lifetime. We consider $100$ randomly distributed jammers and vary the jamming power $P_J$ from $0.1$ to $20$. The graph shows that the jamming power has a great impact on network lifetime. The lifetime increases rapidly when jamming power grows from $0.1$ to about $6$, and maintains stable when jamming power is $6$ to $10$, then decreases when jamming power exceeds $10$. Since both the SINR of eavesdroppers and the SINR of legitimate nodes are determined by the jamming power, when the jamming power increases appropriately, the number of jammers needed to successfully interfere eavesdroppers decreases, consequently network lifetime becomes longer. However, when the jamming power grows too large, it is difficult to guarantee the legitimate nodes from being disturbed, hence the lifetime will suffer a decline.

In Figure \ref{fig:battery}, we measure the network lifetime when the life span of each jammer varies between $1$ and $10$ with an increment of $1$. We consider $100$ jammers randomly deployed. The lifetime increases proportionally to the life span of each jammer and far exceeds the base value .

Figure \ref{fig:delta} shows the relationship between network lifetime and the threshold level $\delta_2$, which indicates how much more capable the eavesdropper is over the legitimate nodes. When $\delta_2$ grows, the eavesdropper becomes less capable, then fewer jammers are needed to secure the communication, thus the network lifetime will increase.

In Figure \ref{fig:sigma}, we study the impact of the percentage of rechargeable jammers $\eta$ on network lifetime. We consider $100$ randomly deployed jammers and vary the percentage of rechargeable jammers between $0$ and $0.8$ with an increment of $0.1$. The rate of energy consumption is set to be $10$ or $20$. The lifetime grows slowly with $\eta$ between $0$ and $0.5$, but increases rapidly when $\eta$ exceeds $0.5$. Under the same percentage of rechargeable jammers, the lifetime with $c=10$ is larger than lifetime with $c=20$. When more rechargeable jammers are deployed, more jammers will be recharged during each time slot, and hence the network lifetime can be prolonged. And smaller rate of consumption means shorter charging period, thus longer lifetime. It is worth to note that, when $\eta$ is greater than a certain threshold, we can find a series of disjoint \emph{reliable} subsets such that when those subsets are activated in a round-robin fashion, the network lifetime is infinite.

In Figure \ref{fig:consume}, we measure the network lifetime when the rate of energy consumption of each jammer, $c$, varies between $4$ and $20$ with an increment of $2$. We consider $100$ randomly deployed jammers and the percentage of rechargeable jammers is $0.1$ or $0.5$. Note that to avoid the impact of life span, we assume the life span of each jammer is the constant $10$. The network lifetime decreases rapidly when $c$ is small and remains stable when $c$ is large enough. When the energy consumption rate is small and the percentage of rechargeable jammers is large, the energy consumed during one time slot will be refilled quickly, thus the lifetime will be longer, even infinite.

\section{Conclusion}
Wireless communication systems are increasingly adopted to transfer potential highly sensitive information, but are easy to be cracked or eavesdropped by adversaries due to the shared nature of wireless medium. Adding artificial noises by friendly jammers is a feasible way to protect the communication systems. This paper studies the schedule strategies of randomly deployed friendly jammers, which can be unrechargeable or rechargeable, to maximize the lifetime of the jammer networks and prevent the cracking of jamming effect made by the eavesdroppers. An ILP-based approximation algorithm is first proposed as baseline, then a heuristic algorithm based on the greedy strategy that less consumption leads to longer lifetime is also proposed with lower complexity. The theoretical analysis and extensive simulations show that our algorithms are effective and efficient. This work may be extended to randomly schedule the activities of jammers in the future.

%In this paper, we have studied the schedule strategies of friendly jammers, including rechargeable and unrechargeable jammers, to maximize the survival time of the jammer network in a circumscribed geographical region. The problem is formulated as a maximization problem under the constraints of geographical area, energy consumption, transmission power and threshold level. We first present an approximation algorithm as baseline using the integer linear programming model. To further reduce the computational complexity, we then propose a heuristic algorithm based on the greedy strategy---less consumption, longer lifetime. Furthermore, a pruning strategy was proposed to reduce the computational steps in our algorithms. Finally, extensive simulation results demonstrate that the proposed algorithms are feasible and efficient.

\bibliographystyle{IEEEtran}
\bibliography{Privacybib}

\end{document}